\documentclass[12pt,oneside]{amsart}
\usepackage[margin=1.0in]{geometry}
\usepackage{amsmath,amsthm,amssymb,amsfonts}
\usepackage{graphicx}
\usepackage{tikz}
\usepackage{latexsym}
\usepackage{caption}
\usepackage{subcaption}
\usepackage[T2A]{fontenc}
\usepackage[utf8]{inputenc}
\usepackage{booktabs}
\usepackage{longtable}
\usepackage[font=small,labelfont=bf]{caption}
\usepackage[russian,english]{babel}
\usepackage{natbib}
\usepackage{ragged2e}
\usepackage{setspace}
\usepackage{xcolor}
\usepackage{csquotes}
\usepackage{multirow}
\usepackage{pdfpages}
\usepackage{hyperref}
 \hypersetup{
    colorlinks   = true,
    citecolor    = blue,
    linkcolor=blue
}
\newtheorem*{proposition*}{Proposition}
\newtheorem*{corollary*}{Corollary}

\newtheorem{proposition}{Proposition}
\newtheorem{corollary}{Corollary}
\newtheorem{lemma}{Lemma}
\newtheorem*{remark*}{Remark}
\newtheorem{definition}{Definition}
\newtheorem{theorem}{Theorem}
\newtheorem{example}{Example}
\newtheorem{claim}{Claim}
\newtheorem*{claim*}{Claim}

\usepackage[foot]{amsaddr}

\makeatletter
\renewcommand{\email}[2][]{%
  \ifx\emails\@empty\relax\else{\g@addto@macro\emails{,\space}}\fi%
  \@ifnotempty{#1}{\g@addto@macro\emails{\textrm{(#1)}\space}}%
  \g@addto@macro\emails{#2}%
}
\makeatother

\begin{document}

\thispagestyle{empty}

\title[]{When do Reforms meet fairness concerns in school admissions?}

\author{Somouaoga Bonkoungou$^{\ast}$}
\address{$^{\ast}$Faculty of Business and Economics, University of
Lausanne, Internef, CH-1015 Lausanne, Switzerland, \textup{bkgsom@gmail.com} \newline $^{\dagger}$Department of Economics and Game Theory Lab, Higher School of Economics, 16, Soyuza Pechatnikov
st., 190121, St. Petersburg, Russia, \textup{nesterovu@gmail.com}}

\author{Alexander Nesterov$^{\dagger}$}

\thanks{We are grateful to Bettina Klaus, Battal Do\u{g}an, Lars Ehlers, and Rustamdjan Hakimov for their suggestions. We thank Camille Terrier, Inacio Bo, Madhav Raghavan, and other participants of the online seminar of the Lausanne Market Design group for their comments, as well as the participants of the online Conference on Mechanism and Institution Design, and the 2022 EEA-ESEM meeting in Milano.\\ 
Acknowledgments: Somouaoga’s work in this project has received funding from the European Union’s Horizon 2020 research and innovation programme under the Marie Skłodowska-Curie grant agreement No 890648. The results of this publication reflect only the authors’ view and the commission is not responsible for any use that may be made to the information it contains. Support from the Basic Research Program and the “Priority 2030” Program of the National Research University Higher School of Economics is gratefully acknowledged. 
}

\date{April 15, 2024}

\begin{abstract}

We study a series of reforms in school admissions mechanisms motivated, among other reasons, by fairness concerns and vulnerability to manipulation. Before the reforms as well as after, the mechanisms were vulnerable to preference manipulation and induced \textit{blocking students}: students who miss desired schools despite having higher priority or seats thereat left empty. We demonstrate that some of these reforms improved fairness by adopting mechanisms with fewer blocking students compared to the preexisting ones, while several others did not. We identify preexisting mechanisms where fairness consideration was more of an issue than vulnerability to manipulation and those where it is the reverse.

$ $
%
%
\end{abstract}

\maketitle

\thispagestyle{empty}

\section{Introduction}

In the last two decades, there has been a wave of reforms of school admissions mechanisms around the world \citep{pathak2013}. The surprising fact is that after such reforms most matching mechanisms present the same poor properties that could have arguably justified the policy changes. For example, despite some evidence that vulnerability to manipulation and fairness concerns mostly drove the changes, most newly adopted matching mechanisms still suffer from these two deficiencies. 


Fairness is at the forefront of the concerns that led to the policy changes. The most vivid example is, perhaps, the 2007 major reform in England, which covers 146 local school admissions systems. According to the then-Secretary of State, Alan Johnson, the reform aimed to \textit{``ensure that admission authorities -- whether local authorities or schools -- operate in a fair way''}  \citep{code2007}. Among other things, the reform prohibited the practice of giving \textit{``priority to children according to the order of other schools named as preference by their parents,''} known as the \textit{first-preference-first principle}. This principle requires that a student who ranks a school higher in her list receives a higher admission priority at this school compared to students who rank it lower. Before the reform, as many as one-third of schools in England used this principle. 

In 2009, Chicago education authorities reformed their Selective High School admission system. They replaced the so-called Boston mechanism that used the first-preference-first principle for each school, arguing that, due to this principle \textit{``high-scoring kids were being rejected simply because of the order in which they listed their [schools] preferences''} \citep{pathak2013}. 
The same Boston mechanism has also been used for college admissions in several provinces in China, and it raised similar complaints. For example, one parent said: \textit{``My child has been among the best students in his school and school district. He achieved a score of 632 in the college entrance exam last year. Unfortunately, he was not accepted by his first choice. After his first choice rejected him, his second and third choices were already full. My child had no choice but to repeat his senior year''} \citep{chenandkesten2017a,nie2007game}. In 2003, more than 3 million students, representing half of the annual intake, were matched to significantly worse colleges than what their grades allowed \citep{wu2020matching}.


These examples illustrate unfairness concerns with the old mechanisms: they can induce a matching with a so-called \textit{blocking student}, that is, a student who missed a school while at least one seat at that school has been assigned to a student with a lower grade or priority or even left empty. The blocking student is entitled to this seat, yet she has not been assigned to it. It is important to note that we define the concept of fairness concerning true preferences and not reported preferences. A matching with no blocking student is \textit{stable} and is viewed as a fair outcome as it eliminates \enquote{justified envy}, a situation in which a student prefers a school that is assigned to another student with lower admission priority \citep{abdul2003}.\footnote{In general, the relation between stability and fairness is more nuanced, see \cite{romm2020stability}.} \cite{galeandshapley} show that for any instance there is a student-optimal stable matching, a matching that every student finds at least as good as any other stable matching. This stable matching can be reached by the student-proposing deferred acceptance algorithm by \cite{galeandshapley}. We refer to it as Gale-Shapley mechanism.

Apart from the first preference first principle, many mechanisms induce blocking students precisely because they have ranking constraints. In such a mechanism each student is allowed to rank-list only a limited number of schools, typically between 3 and 5 \citep{pathak2013}. Even in New York City, where the ranking constraint is 12 and there are more than 1700 schools, around 25\% of students report a complete list of 12 schools, while only 5\% report 9, 10, or 11 schools, suggesting that around 20\% of students in New York City could not list all acceptable schools \citep{abdul2009}. 
Students who missed all their listed schools but could have been admitted to unlisted schools will be dissatisfied with the admissions system and deem it unfair. We consider all blocking students, whether it concerns listed schools for which admissions authorities can verify priority violations or unlisted schools that lead to dissatisfaction (see \citeauthor{calsamiglia2010constrained}, \citeyear{calsamiglia2010constrained}). 

Our first goal is to investigate whether the reforms led to \textit{more fair} matching mechanisms. Our second goal is to investigate the relative importance of vulnerability to manipulation and fairness concerns in the preexisting mechanisms. \cite{chenandkesten2017a} propose to compare mechanisms by set inclusion of problems where they produce stable outcomes. However, this notion cannot further distinguish mechanisms in each instance where they are not stable. A finer and complementary notion is useful, in particular, when the compared mechanisms are not stable or a large fraction of instances or real-life instances lie in this domain. We indeed illustrate that for many of our compared mechanisms, real-life instances are likely to be in a domain where the compared mechanisms are not stable. 
 
To address this problem, we count and compare the number of blocking students across mechanisms. In an instance where two mechanisms are not stable, they can still be contrasted using the number of blocking students.\footnote{To our knowledge, this criterion has been first used by \cite{roth1997turnaround}. \cite{niederle2009market}, \cite{eriksson2008instability} and \cite{dougan2021minimally} used the criterion of counting the number of blocking pairs, which does not allow comparisons in our setting (see Remark in section \ref{sec-comparison}). \cite{dougan2021minimally} and \cite{dogan2022} introduced criteria for counting the number of blocking pairs and blocking students.} Our investigation led to a result that supports an important kind of reform. Broadly, these reforms involve extending ranking constraints in the Gale-Shapley mechanism. The Gale-Shapley mechanism with a relaxed constraint has weakly fewer blocking students than the restricted counterpart and there are instances where it has fewer blocking students. This took place in Chicago (2010), in Ghana (2007, 2008), in Newcastle (2010), and in Surrey (2010) \citep{pathak2013}. For the remaining reforms, it is not possible to conclude by comparing the number of blocking students. We show that after those reforms the number of blocking students may increase.


We now answer the following question. Was fairness more of a concern compared to vulnerability to manipulation? We focus on blocking students and manipulating students, students who could gain by misreporting their preferences while others are truthful. We show that for any instance the constrained Gale-Shapley mechanism has weakly more blocking students than manipulating students. More precisely, any manipulating student is a blocking student of the mechanism. In contrast, for any instance, the constrained Boston mechanism has weakly more manipulating students than blocking students. More precisely, any blocking student is a manipulating student of the mechanism. For the constrained serial dictatorship mechanism any blocking student is a manipulating student and vice versa. 
A more subtle relationship between stability and manipulability can be seen in the reform in England. After this reform, the mechanisms in most school districts did not become less manipulable \citep{bonkoungou2020} and they did not become fairer by stability either (Example \ref{ex-fpf}). However, the reform was successful according to at least one dimension by the following criterion: if the reform disrupted fairness --- by producing an unstable matching while it was stable before the reform --- the new mechanism is not vulnerable to manipulation. 

\textbf{Related literature.} 
Apart from papers studying the reforms mentioned earlier \citep{pathak2013,chenandkesten2017a,bonkoungou2020,bonkoungou2021incentives,imamura2022measuring} and papers that used the method of counting blocking agents and blocking pairs \citep{roth1997turnaround,niederle2009market,eriksson2008instability} there is recent literature interested in various ways of comparing matching mechanisms by fairness. 

Among strategy-proof and Pareto efficient mechanisms, the Gale's Top Trading Cycles mechanism \citep{shapley1974} is among the most fair by stability when each school has one seat \citep{abdulkadiroglu2019efficiency}. This result also holds for other fairness comparisons, such as the set of blocking students \citep{dogan2022} and the set of blocking triplets $(i,j,s)$ -- student $i$ blocking the matching of school $s$ with student $j$ \citep{kwon2019justified}. The result holds for any stability comparison that satisfies a few basic properties \citep{dogan2022}.

Among Pareto efficient mechanisms, the Efficiency Adjusted Deferred Acceptance mechanism (EADA) due to \cite{kesten2010} is among the most fair in terms of blocking pairs and blocking triplets \citep{dougan2021minimally,tang2021weak,kwon2019justified}. Independent from the present work, \cite{dougan2021minimally} also introduced the fairness comparison by counting to show that among efficient mechanisms, EADA is not the most fair by counting unless the priority profile satisfies a few acyclicity conditions. 

The first papers that studied constrained mechanisms are \cite{romero1998implementation} and \cite{haeringer2009constrained}. They study the stability of Nash equilibrium outcomes of the game induced by these mechanisms. The most important insight is that the Nash equilibrium outcomes of the constrained Boston mechanism are all stable, while the Nash equilibrium outcomes of the constrained Gale-Shapley may not all be stable.\footnote{\cite{erginandsonmez} showed that the Nash equilibrium outcomes of the unconstrained Boston mechanism are stable.} Besides, the Nash equilibrium outcomes of the constrained Gale-Shapley are a subset of the Nash equilibrium outcomes of any constrained Gale-Shapley with a relaxed constraint. Therefore, when the Nash equilibrium outcomes of the constrained Gale-Shapley with a relaxed constraint are all stable, the Nash equilibrium outcomes of the constrained Gale-Shapley with a restricted constraint are also stable.

The rest of the paper is organized as follows. In Section 2, we introduce the model. In Section 3, we present fairness comparisons. In Section 4, we study the relationship between stability and manipulability. We present most of the proofs in the Appendix.

\vspace{0.3cm}
\section{Model}

We consider the school choice problem \citep{balinski1999,abdul2003}. It consists of the following elements: 
\begin{itemize}
    \item a finite set $I$ of students,
    \item a finite set $S$ of schools,
    \item a profile $P=(P_{i})_{\in I}$ of preference relations for each student,
    \item a profile $\succ=(\succ_s)_{s\in S}$ of priority orders for each school, and
    \item a vector $q=(q_{s})_{s\in S}$ of capacities for each school
\end{itemize}
where $P$ and $\succ$ are defined as follows.
For each student $ i$, $ P_i$ is a strict preference relation $ P_{i}$ over $ S\cup \{\emptyset\}$, where $ \emptyset$ represents the outside option of being unmatched. For each school $ s$, $ \succ_s$ is a strict priority order over $ I$. For each student $ i$, let $ R_{i}$ denote the ``at least as good as'' relation associated with $ P_{i}$.\footnote{That is, for each $s,s'\in S\cup\{\emptyset\} $, $ s\mathrel{R_{i}}s'$ if and only $ s\mathrel{P_{i}}s'$ or $ s=s'$.} School $ s$ is \textbf{acceptable} to student $ i$ if $ s \mathrel{P_{i}} \emptyset$; and it is \textbf{unacceptable} to student $ i$ if $ \emptyset\mathrel{P_{i}}s$. We extend the priority order $ \succ_{s}$ of each school $ s$ over $ I$ to the set $ 2^{I}$ of subsets of students and assume that it is responsive to the priority order over $ I$ \citep{roth1985col}. By definition, the priority order $ \succ_{s}$ over $ 2^{I}$ is responsive if for any students $ i,j\in I$ and any subset $ N\subset I\setminus\{i,j\}$ such that $ \lvert N\lvert < q_{s}$, (i) $ N\cup\{i\} \mathrel{\succ_{s}} N$, and (ii) $ N\cup\{i\} \mathrel{\succ_{s}} N\cup\{j\}$ if and only if $ i\mathrel{\succ_{s}}j$. Let $ k\in \{1,\hdots, \lvert S\lvert\}$, and $P_i$ a preference relation where student $ i$ has $x$ acceptable schools. The truncation of $P_i$ after the $ k$'th acceptable school (if any) of the preference relation $ P_{i}$ is a preference relation with $ \min(x,k)$ acceptable schools such that all schools are ordered as in $ P_{i}$. Let $ P^{k}_{i}$ denote the truncation of $ P_{i}$ after the $k$'th acceptable school. Let $ P^{k}=(P^{k}_{i})_{i\in I}$. Given a proper subset $I'\subsetneq I$ of students, we will often write a preference profile as $ P=(P_{I'}, P_{-I'})$ to emphasize the components for students in $ I'$. The tuple $ (I, S, P,\succ,q)$ is a school choice problem or simply a \textbf{problem}. We assume that there are more students than schools, that is, $ \lvert I\lvert > \lvert S\lvert $. The set of students and the set of schools are fixed throughout the paper, and we denote the (school choice) problem by the triple $ (P,\succ,q)$.

A \textbf{matching} $ \mu$ is a function $ \mu:I \rightarrow S\cup \{\emptyset\}$ such that for each school $ s$, $\lvert \mu^{-1}(s)\lvert \leq q_{s}$. We say that student $ i$ is matched under $ \mu$ if $ \mu(i)\neq\emptyset$ and unmatched under $ \mu$ if $ \mu(i)=\emptyset$. Let $ (P,\succ,q)$ be a problem. A matching $ \mu$ is \textbf{individually rational} under $ P$ if for each student $ i$, $ \mu(i) \mathrel{R_{i}} \emptyset$. A pair $ (i,s)$ of a student and a school \textbf{blocks} the matching $ \mu$ under $ (P,\succ,q)$ if $ s\mathrel{P_{i}}\mu(i)$ and either there is a student $ j$ such that $ \mu(j)=s$ and $i \mathrel{\succ_{s}}j$ or $ \lvert\mu^{-1}(s)\lvert <q_{s}$. Student $ i$ is a \textbf{blocking student} for the matching $ \mu$ under $ (P,\succ,q)$ if there is a school $ s$ such that the pair $ (i,s)$ blocks $ \mu$ under $ (P,\succ,q)$. A matching $ \mu$ is \textbf{stable} at $ (P,\succ,q)$ if it is individually rational under $ P$ and has no blocking student. We often drop the problem and refer to a stable matching. A \textbf{mechanism} $ \varphi$ is a function that maps each problem to a matching. For each problem $ (P,\succ,q)$, let $ \varphi_{i}(P,\succ,q)$ denote the component for student $ i$. A mechanism $ \varphi$ is individually rational if for each problem $ (P,\succ,q)$ the matching $ \varphi(P,\succ,q)$ is individually rational under $ P$. A mechanism $ \varphi$ is stable if for each problem $ (P,\succ,q)$ the matching $ \varphi(P,\succ,q)$ is stable at $ (P,\succ,q)$. We often drop the problem and say that a mechanism is stable (at the implicitly assumed problem).

\subsection{Mechanisms} We are interested in mechanisms that were used either before or after the reforms. We first describe unconstrained versions.
\subsubsection*{Gale-Shapley}
\cite{galeandshapley} showed that for each problem, there exists a stable matching. In addition, there is a student-optimal stable matching, which is a matching that each student finds at least as good as any other stable matching. For each problem $ (P,\succ,q)$, this matching can be found via the \cite{galeandshapley} student-proposing deferred acceptance algorithm.

\begin{itemize}
\item \underline{\textit{Step 1:}} Each student applies to her most preferred acceptable school (if any). If a student did not rank any school acceptable, then she remains unmatched. Each school $ s$ considers its applicants at the first step denoted as  $ I^{1}_{s}$ and tentatively accepts $\min(q_{s},\lvert I^{1}_{s}\lvert) $ of the $ \succ_{s}$-highest priority applicants and rejects the remaining ones. Let $ A^{1}_{s}$ denote the set of students whom school $ s$ has tentatively accepted at this step.
\item \underline{\textit{Step t>1:}} Each student, who is rejected at step $ t-1$, applies to her most preferred acceptable school among those which have not yet rejected her (if any). If a student does not have any remaining acceptable school, then she remains unmatched. Each school $ s$ considers the set $A^{t-1}_{s}\cup I^{t}_{s}$, where $ I^{t}_{s}$ are its new applicants at this step, and tentatively accepts $ \min(q_{s},\lvert A^{t-1}_{s}\cup I^{t}_{s}\lvert)$ of the $ \succ_{s}$-highest priority applicants and rejects the remaining ones. Let $ A^{t}_{s}$ denote the set of students whose school $ s$ has tentatively accepted at this step.
\end{itemize}
The algorithm stops when no student is rejected. The tentative acceptances become final at this step. Let $ GS$ denote this mechanism. Given $ k\in \{1,...,\lvert S\lvert\}$, the constrained version $ GS^{k}$ of the Gale-Shapley mechanism $ GS$ is the mechanism that assigns to each problem $ (P,\succ,q)$ the matching $ GS(P^{k},\succ,q)$. That is, $ GS^{k}(P,\succ,q)=GS(P^{k},\succ,q)$. 

\subsubsection*{Serial Dictatorship} When schools have the same priority order, we call the Gale-Shapley mechanism the serial dictatorship mechanism.\footnote{This is a slight abuse of our definition since the domain of a mechanism is the set of all problems --- including problems where schools have different priorities.} Let $ SD$ denote this mechanism. The outcome of this mechanism can be computed via the following simplified process. Students move in sequence following the common priority order. The first student picks her most preferred acceptable school. The next student picks her most preferred acceptable school among the remaining ones, and so on. Given $ k\in \{1,...,\lvert S\lvert\}$, the constrained version $ SD^{k}$ of the Serial Dictatorship mechanism $ SD$ is the mechanism that assigns to each problem $ (P,\succ,q)$ the matching $ SD(P^{k},\succ,q)$. That is, $ SD^{k}(P,\succ,q)=SD(P^{k},\succ,q)$.

\subsubsection*{First-Preference-First} The schools are exogenously divided into two disjoint subsets $ S^{fpf}$ and $ S^{ep}$ such that $ S^{fpf}\cup S^{ep}=S$. The set $ S^{eq}$ is a set of \textbf{equal-preference schools} and $ S^{fpf}$ is a set of \textbf{first-preference-first} schools. 
The First-Preference-First mechanism (FPF) assigns to each problem $ (P,\succ,q)$, the matching $ GS(P,\hat{\succ},q)$ where $\hat{\succ}$ is obtained as follows. The priority order of each equal-preference school is maintained intact while the priority order of each first-preference-first school is adjusted according to the rank that students have assigned to it. Formally, the priority profile $ \hat{\succ}$ is obtained as follows:
\par 1.~for each equal-preference school $ s\in S^{ep}$, $ \hat{\succ}_{s}=\,\succ_{s}$ and
\par 2.~for each first-preference-first school $ s\in S^{fpf}$, $ \hat{\succ}_{s}$ is defined as follows. Let $ I^{1}(s)$ be the set of students who have ranked school $ s$ first under $ P$, $ I^{2}(s)$ the set of students who have ranked school $ s$ second under $ P$, and so on. Note that we count the ranking of $ \emptyset$ as well.
\begin{itemize}
\item For each $ \ell,k\in \{1,\hdots,\lvert S\lvert+1\}$ such that $ \ell>k$ and each students $ i,j$ such that $i\in I^{k}(s)$ and $ j\in I^{\ell}(s)$, $ i\mathrel{\hat{\succ}_{s}}j$. 

\item For each $ k\in \{1,\hdots,\lvert S\lvert+1\}$ and each $ i,j\in I^{k}(s)$, $ i\mathrel{\hat{\succ}_{s}}j$ if and only if $ i\mathrel{\succ_{s}}j$. 
\end{itemize}
Let $ FPF$ denote this mechanism. Given $ k\in \{1,...,\lvert S\lvert\}$, the constrained version $ FPF^{k}$ of the First-Preference-First mechanism $ FPF$ is the mechanism that assigns to each problem $ (P,\succ,q)$ the matching $ FPF(P^{k},\succ,q)$. That is, $ FPF^{k}(P,\succ,q)=FPF(P^{k},\succ,q)$.

\subsubsection*{Boston} Until 2005, the Boston public school system was using an immediate acceptance mechanism called the Boston mechanism \citep{abdul2003}. This mechanism assigns to each problem $ (P,\succ,q)$, the matching as described in the following algorithm.

\begin{itemize}
\item \underline{\textit{Step 1:}} Each student applies to her most preferred acceptable school (if any). Each school $ s$, considers its applicants at the first step denoted as $ I^{1}_{s}$ and immediately accepts $\min(q_{s},\lvert I^{1}_{s}\lvert) $ of the $ \succ_{s}$-highest priority applicants and rejects the remaining ones. For each school $ s$, let $ q^{1}_{s}=q_{s}-\min(q_{s},\lvert I^{1}_{s}\lvert)$ denote its remaining capacity after this step. 

\item \underline{\textit{Step t>1:}}  Each student who is rejected at step $ t-1$, applies to her most-preferred acceptable school among those who have not yet rejected her (if any). Each school $ s$ considers its new applicants $ I^{t}_{s}$ at this step and immediately accepts $ \min(q^{t-1}_{s},\lvert I^{t}_{s}\lvert)$ of the $ \succ_{s}$-highest priority applicants and rejects the remaining ones. For each school $ s$, let $ q^{t}_{s}=q^{t-1}-\min(q^{t-1}_{s},\lvert I^{t}_{s}\lvert)$ denote its remaining capacity after this step.
\end{itemize}

The algorithm stops when every student is either accepted at some step or has applied to all of her acceptable schools. Let $ \beta$ denote this mechanism. Given $ k\in \{1,...,\lvert S\lvert\}$, the constrained version $ \beta^{k}$ of the Boston mechanism $ \beta$ is the mechanism that assigns to each problem $ (P,\succ,q)$ the matching $ \beta(P^{k},\succ,q)$. That is, $ \beta^{k}(P,\succ,q)=\beta(P^{k},\succ,q)$.

\begin{remark*}
In the (algorithm of the) Boston mechanism, students applying to the same school at each step have assigned the same rank to it. Therefore, students applying to a school at a given step of the algorithm rank this school higher than those applying to it at any step after. In particular, no student could be rejected by a school while another student, who has assigned a lower rank to it, is accepted by this school. Thus, the Boston mechanism is a First-Preference-First mechanism where every school is a first-preference-first school. This result follows from the Proposition 2 of \cite{pathak2008}.
\end{remark*}


\subsubsection*{Chinese parallel} \cite{chenandkesten2017a} describe a parametric mechanism that many Chinese provinces have been using. The parameter $ e\geq 1$ is a natural number. For each problem $(P,\succ,q) $, the outcome is a sequential application of constrained $GS$. In the first round, the matching is final for students who are matched under $ GS^{e}(P,\succ,q)$, while unmatched students proceed to the next round. In the next round, each school reduces its capacity by the number of students assigned to it in the last round, each matched student replaces her preferences with a preference relation where she finds no school acceptable and the unmatched students (in the previous round) are matched according to $GS^{2e}$ for the reduced capacities and the new preference profile. The process continues until either no school has a remaining seat or no unmatched student finds a school with a remaining seat acceptable. Let $ Ch^{(e)}$ denote this mechanism.\footnote{This definition of the Chinese parallel mechanisms is given only for the symmetric version where each round has the same length $e$. See \cite{chenandkesten2017a} for details.}


\section{Comparison of Mechanisms}\label{sec-comparison}

In this section we compare mechanisms according to two criteria: fairness by stability and fairness by counting.

\subsection{Fairness by stability}
Our starting point is a comparison according to the set inclusion of problems where mechanisms are stable.

\begin{definition} \citep{chenandkesten2017a}. Mechanism $\varphi'$ is \textbf{more fair by stability} than $ \varphi$ if (i) at each problem where $\varphi$ is stable, $\varphi'$ is also stable and (ii) there exists a problem where $\varphi'$ is stable but $\varphi$ is not.
\end{definition}

This criterion is less demanding in the sense that it does not take into account problems where mechanisms produce unstable outcomes. However, it does not explain many changes that followed the 2007 reform in the UK as the constrained First-Preference-First mechanism is not comparable to the constrained Gale-Shapley mechanism according to this criterion. We demonstrate this in the following example.\footnote{Example where $FPF^k$ has blocking students while $GS^k$ does not is immediate and omitted.} 

\begin{example} \label{ex-fpf}

Let $I=\{i_{1},\hdots,i_{7}\}$ and $S=\{s_{1},\hdots,s_{5}\}$. Let school $s_3$ be the only first-preference-first school. Let $(P,\succ,q)$ be a problem where each school has one seat and the remaining components are specified as follows. (The sign $ \vdots$ indicates that the remaining part is arbitrary.)

\begin{center}
\begin{tabular}{@{}lllllll|lllll@{}}
$P_{i_1}$& $P_{i_2}$   & $P_{i_3}$   & $P_{i_4}$   & $P_{i_5}$   & $P_{i_6}$   & $P_{i_7}$   & $\succ_{s_1}$& $\succ_{s_2}$& $\succ_{s_3}$& $\succ_{s_4}$& $\succ_{s_5}$\\ \midrule
$s_1$& $s_1$   & $s_4$   & $s_1$   & $s_2$   & $s_1$   & $s_5$   & $i_4$& $i_5$& $i_3$& $i_1$& $i_7$\\
$s_2$& $s_3$   & $s_3$   & $s_2$   & $s_1$   & $s_2$   & $s_1$   & $\vdots$ & $\vdots$ & $i_1$& $i_6$& $\vdots$ \\
$s_3$& $\emptyset$ & $\emptyset$ & $s_3$   & $s_3$   & $s_5$   & $s_2$   &  &  & $i_2$& $i_3$&  \\
$s_4$& & & $\emptyset$ & $\emptyset$ & $s_3$   & $\emptyset$ &  &  & $\vdots$ & $\vdots$ &  \\
$\emptyset$ & & & & & $s_4$   & &  &  &  &  &  \\
 & & & & & $\emptyset$ & &  &  &  &  & 
\end{tabular}
\end{center}

The outcomes of the constrained First-Preference-First $FPF^4$ and the constrained Gale-Shapley $GS^4$ at $(P,\succ,q)$ are as follows: 

\begin{equation*}
FPF^{4}(P,\succ,q)=\begin{pmatrix}
i_1 & i_2 & i_3 & i_4 & i_5 & i_6 & i_7 \\
s_4 & \emptyset & s_3 & s_1 & s_2 & \emptyset & s_5
\end{pmatrix},
\end{equation*}

\begin{equation*}
GS^{4}(P,\succ,q)=\begin{pmatrix}
i_1 & i_2 & i_3 & i_4 & i_5 & i_6 & i_7 \\
s_3 & \emptyset & s_4 & s_1 & s_2 & \emptyset & s_5
\end{pmatrix}.
\end{equation*}
The matching $FPF^{4}(P,\succ,q)$ is stable.\footnote{This matching is both the student-optimal and the school-optimal stable matching.} However, the matching $GS^{4}(P,\succ,q)$ is not stable. Indeed, the pair $(i_6,s_{4})$ blocks this matching because student $ i_{6}$ is unmatched and finds school $s_4$ acceptable, but student $ i_{3}$ is matched to $s_{4}$ while $ i_{6}\mathrel{\succ_{s_{4}}}i_{3}$. 
The intuition is that the constraint in $ GS$ shortened the chains of the rejections needed to reach a stable matching in the Gale-Shapley algorithm. For example, student $ i_{3}$ is temporarily matched to school $ s_{4}$ at some step of the algorithm. At the student-optimal stable matching for $ (P,\succ,q)$, school $ s_{4}$ is assigned to student $ i_{1}$. However, we need an application of student $ i_{1}$ at that school to displace student $ i_{3}$ from $ s_{4}$. This does not occur under $ GS^{4}$ because no student initiates the rejection chain. However, under $ FPF^{4}$, the application of student $ i_{2}$ at school $ s_{3}$ causes the rejection of student $ i_{1}$ at $ s_{3}$ (student $ i_{2}$ has ranked it higher than $ i_{1}$ and school $ s_{3}$ is a first-preference-first school). This is the rejection needed to reach the student-optimal stable matching.

\end{example}
This example illustrates how the constrained $GS$ mechanism has shortened the chains needed to reach a stable matching. It is well known that this type of chain leads to unambiguous welfare losses. Each student in the chain is worse off, and all other students are unaffected \citep{kesten2010}.\footnote{These chains are initiated by the so-called interrupters. These are students who initiate chains of rejections that return to them \citep{kesten2010}.} However, under the Boston mechanism, (where all schools are first-preference-first schools) there is no such chain, and thus the constrained Boston mechanism can be compared to the constrained Gale-Shapley mechanism. We also show the comparison result for the Gale-Shapley mechanisms with different constraints.

\begin{proposition}

 \label{thm1} Let $ \ell$ and $k$ be positive integers such that $k>\ell> 1$ and suppose that there are at least two schools:

(i) The constrained Gale-Shapley mechanism $GS^{k}$ is more fair by stability than the constrained Boston mechanism $ \beta^{k}$.
 
(ii) The constrained Gale-Shapley mechanism $ GS^{k}$ is more fair by stability than constrained Gale-Shapley mechanism $ GS^{\ell}$.
 
\end{proposition}

 \cite{chenandkesten2017a} have established that any (unconstrained) Chinese mechanism $Ch^e$ is more stable than any Chinese mechanism $Ch^{e'}$ where $e'=k e$ for $k\in \mathbb{N}\cup \{\infty\}$. Their result and ours are similar but not corollary of each other. Indeed, $Ch^1$ is the Boston mechanism and $Ch^{\infty}$ is the Gale-Shapley mechanism such that for a problem $(P,\succ,q)$, we can write $\beta^k(P,\succ,q)=Ch^{(1)}(P^k,\succ,q)$ and $GS^k(P,\succ,q)=Ch^{\infty}(P^k,\succ,q)$. Our results concern constrained Chinese mechanisms where both the parameter $e$ and the constraint $k$ could be a source of blocking while in \cite{chenandkesten2017a}, the parameter $e$ is the only source of blocking. 

The difficulty with fairness by stability is that, for problems we encounter in real life, all mechanisms described above are likely to induce unstable outcomes, and the comparisons therefore are driven by some less relevant problems. Consider the high school admissions problem in Chicago where schools have a common priority (constructed from student's composite scores) and where students form block preferences as illustrated in the following example.    

\begin{example}[Tier preferences]
Consider a problem with $n$ students and $m$ schools such that for each $s,s'\in S$, $\succ_s\mathrel{=}\succ_{s'}$. 
We assume that students have tier preferences. The set $ S$ of schools is partitioned into two sets $S_{1},S_{2}$. Each student $i$ prefers each school in $ S_{1}$ to each school in $S_{2}$.\footnote{\cite{coles2013} observed that the academic job market has this structure and referred to it as block-correlated preferences.} We assume that each student finds each school acceptable and $n>\sum_{s}q_{s}$.

Whenever $|S_{1}|\geq k$, no student ranks a school in $ S_{2}$ among the top $ k$ acceptable schools. Any individually rational and $k$-constrained ranking mechanism has a blocking student. Indeed, if every student reports her preferences truthfully, then some students are unmatched while seats at schools in $S_{2}$ are unassigned. 

But we can still distinguish different constrained Gale-Shapley mechanisms by the number of blocking students within this preference domain. For example, suppose that students have the same ranking over all schools in $S_1$ and we compare $GS^{k-1}$ and $GS^{k}$. Every unmatched student is a blocking student (and vice versa), and $GS^{k}$ matches strictly more students than $GS^{k-1}$ and thus $GS^{k}$ has fewer blocking students than $GS^{k-1}$ does.\end{example}

While this example is simplified to illustrate the fact that fairness by stability might not be able to distinguish the mechanisms that we study in real-life instances, all that is necessary for the result is that many students rank a subset of schools such that they exhaust the listing constraint, and such that these schools do not have enough seats to accommodate all these students. 

\subsection{Fairness by counting}

In this section, we introduce a criterion for comparing the number of blocking students. With this criterion, mechanisms can be compared at problems where both induce unstable outcomes.

\begin{definition} A mechanism $ \varphi'$ is \textbf{more fair by counting} (the number of blocking students) than a mechanism $ \varphi$ if 
(i) for each problem, there are at least as many blocking students of the outcome of $\varphi$ as there are of the outcome of $\varphi'$, and
(ii) there is a problem where there are more blocking students of the outcome of $ \varphi$ than the outcome of $ \varphi'$.

\end{definition}

Fairness by counting is stronger than fairness by stability considered earlier. If a mechanism $\varphi'$ is more fair by counting than $\varphi$, then for each problem where $\varphi$ induces a stable matching, i.e., there is no blocking student, $\varphi'$ also necessarily induces a stable matching. Our main result with this concept is a strengthening of the comparison between different constraints of the Gale-Shapley mechanism. We illustrate the intuition using the example below.

\begin{example} \label{ex-GS2to1}

Let $I=\{i_{1},\hdots,i_{5}\}$ and $S=\{s_{1},\hdots,s_{4}\}$. Let $ (P,\succ,q)$ be a problem where each school has one seat, and the remaining components are specified as follows.

\begin{center}
\begin{table}[ht]
\centering
\begin{tabular}{c c c c c|c c c c}
 $ P_{i_{1}}$ & $ P_{i_{2}}$ & $ P_{i_{3}}$ & $ P_{i_{4}}$ & $ P_{i_{5}}$  & $ \succ_{s_{1}}$ & $ \succ_{s_{2}}$ & $ \succ_{s_{3}}$ & $ \succ_{s_{4}}$ \\\hline
 $ s_{1}$ & $ s_{1}$ & $ s_{2}$ & $ s_{3}$ & $s_{3} $ & $ i_{3}$ & $ i_{2}$ & $ i_{1}$ & $ i_{5}$\\
 $ s_{2}$ & $ s_{2}$ & $ s_{1}$ & $ s_{1}$ & $s_{4} $ & $ i_{1}$ & $ i_{4}$ & $ i_{5}$ & $ \vdots$ \\
 $ s_{3}$ & $ s_{3}$ & $ s_{3}$ & $ s_{2}$ & $\vdots $ & $ \vdots$ & $ \vdots$ & $ \vdots$ &
 \end{tabular}
 \label{tab:my_label}
\end{table}
\end{center}
  Let us compare the mechanisms $ GS^{2}$ and $ GS^{1}$. We have
  
 \begin{equation*}
 GS^{2}(P,\succ,q)= \begin{pmatrix} i_{1} & i_{2} & i_{3} & i_{4} & i_{5} \\ \emptyset & s_{2} & s_{1} & \emptyset & s_{3}
 \end{pmatrix}
 \end{equation*}
where student $ i_{1}$ is the unique blocking student for the matching under $ (P,\succ,q)$. Indeed, $i_1$ is unmatched, finds $s_3$ acceptable and has a higher priority at $ s_{3}$ than $ i_{5}$. Let us shorten the reported list only for student $ i_{2}$. Then,
\begin{equation*}
 GS^{2}(P^{1}_{i_{2}},P_{-i_{2}},\succ,q)= \begin{pmatrix} i_{1} & i_{2} & i_{3} & i_{4} & i_{5} \\ s_{1} & \emptyset & s_{2} & \emptyset & s_{3}
 \end{pmatrix}.
 \end{equation*}
 As a result of this replacement, there are three types of students, given their status in the previous matching. First, student $ i_{2}$ --- who was matched --- became a blocking student. Second, student $ i_{1}$ --- who was a blocking student --- is not a blocking student for the new matching. Finally, student $ i_{4}$ is a new blocking student.
 
 The intuition of this result is that by shortening the schools listed by student $ i_{2}$, she is worse off while the other students are weakly better off. First, she is a blocking student for the new matching. Second, student $ i_{1}$ is not a blocking student for the new matching, though she was a blocking student for the old matching. But a new blocking student appears so  there are two blocking students in total. 
\end{example}

This example turns out to be a general pattern. When students shorten the list, the set of blocking students changes, but the size of this set never decreases. By sequentially applying this argument to all students, we get the following result.\footnote{We are very grateful to a referee for suggesting a much simplified technique for proving this result.}

\begin{theorem}
Suppose that there are at least two schools and let $ \lvert S\lvert > k >\ell\geq 1$. The constrained Gale-Shapley mechanism $GS^{k}$ is more fair by counting than the constrained Gale-Shapley mechanism $GS^{\ell}$.\label{Th-GS-counting}
\end{theorem}

Next, we show that this result does not extend to other mechanisms.

\begin{example}[Constrained Boston and constrained Gale-Shapley] \label{example3}

Let $n\geq 7$, $I=\{i_{1},...,i_{n}\}$ and $S=\{s_{1},\hdots, s_{5}\}$. Let $(P,\succ,q) $ be a problem where each school has one seat and the remaining components are specified as follows.

\begin{center}

\begin{tabular}{llllllll|l}
$P_{i_{1}}$& $P_{i_{2}}$& $P_{i_{3}}$& $P_{i_{4}}$& $P_{i_{5}}$  & $\hdots$& $P_{i_{n-1}}$   & $P_{i_{n}}$& $\succ_{s,\, s\in S}$  \\ \hline
$s_1$& $s_2$& $s_3$& $s_1$& $s_1$   & $s_1$   & $s_1$   & $s_4$& $i_1$\\
$\vdots$ & $\vdots$ & $\vdots$ & $s_4$& $s_2$   & $s_2$   & $s_2$   & $s_{5}$  & $i_2$\\
 &  &  & $s_{5}$  & $s_3$   & $s_3$   & $s_3$   & $\vdots$ & $i_3$\\
 &  &  & $\vdots$ & $s_{5}$ & $s_{5}$ & $s_{5}$ &  & $i_4$\\
 &  &  &  & $\emptyset$ & $\emptyset$ & $\emptyset$ &  & $i_5$\\
 &  &  &  & & & &  & $\vdots$ \\
 &  &  &  & & & &  & $i_n$   
\end{tabular}
\end{center}

The outcomes of $ \beta^{3}$ and $ GS^{3}$ for this problem are specified as follows:
\begin{equation*}
\beta^{3}(P,\succ,q)=\begin{pmatrix} i_{1} & i_{2} & i_{3} & i_{4} & i_{5} &\hdots & i_{n-1}& i_{n}\\ s_{1} & s_{2} & s_{3} & s_{5} & \emptyset& \hdots & \emptyset & s_{4}
\end{pmatrix}
\end{equation*} 
and 
\begin{equation*}
GS^{3}(P,\succ,q)=\begin{pmatrix} i_{1} & i_{2} & i_{3} & i_{4} & i_{5} &\hdots & i_{n-1}& i_{n}\\ s_{1} & s_{2} & s_{3} & s_{4} & \emptyset& \hdots & \emptyset & s_{5}
\end{pmatrix}.
\end{equation*}


Let us compare the number of blocking students for the two matchings. On one hand, student $i_4$ is the only blocking student for $ \beta^{3}(P,\succ,q)$. Indeed, the pair $ (i_{4},s_{4})$ blocks $ \beta^{3}(P,\succ,q)$ under $ (P,\succ,q)$. On the other hand, students $ i_{5},\hdots,i_{n-1}$ are all blocking students of $ GS^{3}(P,\succ,q)$ because they are unmatched, each of them prefers school $ s_{5} $ to being unmatched, and has higher priority than $ i_{n}$ under $ \succ_{s_{5}}$. Since $ n\geq 7$, there are at least two blocking students of $GS^{3}(P,\succ,q)$. Therefore, there are more blocking students of $ GS^{3}(P,\succ,q)$ than $ \beta^{3}(P,\succ,q)$. By Theorem \ref{thm1}, there is a problem where $ GS^{3}$ is stable but not $ \beta^{3}$.
\end{example}


\begin{remark*}[Chinese parallel]\label{example4} Consider the Chinese mechanisms $ Ch^{(1)}=\beta$ and $ Ch^{(3)}$ and note that for the problem $ (P,\succ,q)$ specified in Example \ref{example3}, $ Ch^{(1)}(P,\succ,q)=\beta^{3}(P,\succ,q)$ and $ Ch^{(3)}(P,\succ,q)=GS^{3}(P,\succ,q)$. According to the conclusion in Example \ref{example3}, there are more blocking students for $ Ch^{(3)}(P,\succ,q)$ than $ Ch^{(1)}(P,\succ,q)$. According to \cite{chenandkesten2017a}, there is a problem where $ Ch^{(3)}$ produces a stable outcome but $ Ch^{(1)}$ does not.

\end{remark*}



\begin{proposition} There is $k>1$, $e'>e$, and a problem such that, (i) the constrained Gale Shapley mechanism $GS^k$ is not more fair by counting than the constrained Boston mechanism $\beta^k$ and (ii) the Chinese parallel mechanism $Ch^{(e')}$ is not more fair by counting than the Chinese parallel mechanism $Ch^{(e)}$.
\end{proposition}
 
\begin{remark*}
Two other notions, comparing mechanisms by the inclusion of blocking pairs and blocking students, have also been studied by \cite{dougan2021minimally}. However, these criteria are stronger than fairness by counting (if the set of blocking pairs or blocking students shrinks, then the number of blocking students does as well) and will lead to negative results for our comparisons. To see this, consider Example \ref{example3}. In this example, $ (i_{5},s_{5})$ is a blocking pair for $ SD^{4}(P,\succ,q)$ but not for $ \beta^{4}(P,\succ,q)$. In addition, $ (i_{4},s_{4})$ is a blocking pair for $ \beta^{4}(P,\succ,q)$ but not for $ SD^{4}(P,\succ,q)$. 

For the comparison between different constrained Gale-Shapley, consider Example \ref{ex-GS2to1} where $ (i_{1},s_{3})$ is a blocking pair for $ GS^{2}$ but not $ GS^{1}$. In addition, $ (i_{2},s_{2})$ is a blocking pair for $ GS^{1}$ but not $ GS^{2}$.
\end{remark*}

\section{Stability and manipulability} 

In this section, we will elucidate the relation between blocking students and manipulating students, i.e., those who may benefit from misrepresenting their preferences to the mechanisms.  

\begin{definition} Let $ \varphi$ be a mechanism. (i) Student $ i$ is a \textbf{manipulating student} of $ \varphi$ at $ (P,\succ,q)$ if there is a preference relation $ P'_{i}$ such that 
\begin{equation*}
\varphi_{i}(P'_{i},P_{-i},\succ,q) \mathrel{P_{i}} \varphi_{i}(P,\succ,q).
\end{equation*}
(ii) Mechanism $ \varphi$ is \textbf{not manipulable} at $ (P,\succ,q)$ if there is no manipulating student of $ \varphi$ at $ (P,\succ,q)$.
\end{definition}
It turns out that there is a relationship between blocking students and manipulating students for the constrained Boston mechanism and the constrained Gale-Shapley mechanism. Interestingly, these relations between the two mechanisms are reversed.

\begin{theorem} Let $ k>1$. For any problem, (i) every blocking student of the constrained Boston mechanism $ \beta^{k}$ is a manipulating student, (ii) every manipulating student of the constrained Gale-Shapley mechanism $ GS^{k}$ is a blocking student, and (iii) every manipulating student of the constrained serial dictatorship mechanism $SD^k$ is a blocking student and vice versa.\label{theorem5}
\end{theorem} 

Note that there are instances where the reverse of parts (i) and (ii) of the theorem are not true. For part (ii) see Example \ref{example5}. For part (i) consider Example \ref{example3} and suppose that there, student $i_n$ prefers school $s_1$ first, $s_4$ next, and $s_5$ last. The preferences of the remaining students are unchanged. Then under $\beta^3$ student $i_n$ is matched to school $s_5$ and is not a blocking student. However, by ranking school $s_4$ first as in the example, she is matched to it. That is she is manipulating student of $\beta^3$. 

These results have important implications for the relation between manipulability and stability. To see this, suppose that there is no manipulating student for the constrained Boston mechanism $\beta^{k} $. Then, by part (i) of this theorem, there is no blocking student of $ \beta^{k}(P,\succ,q)$. Since $ \beta^{k}$ is individually rational, then $ \beta^{k}(P,\succ,q)$ is stable. Suppose now that there is no blocking student for $ GS^{k}(P,\succ,q)$. Since $ GS^{k}$ is individually rational, this means that $ GS^{k}(P,\succ,q)$ is stable. Then, there is no manipulating student of $ GS^{k}$. We summarize these results in the following corollary and in Figures 1 and 2.



\begin{corollary}
(i) If the constrained Boston mechanism $\beta^{k}$ is not manipulable, then it is stable. (ii) If the constrained Gale-Shapley mechanism $ GS^{k}$ is stable, it is not manipulable. \label{cor-B-GS-stab-manip}

\end{corollary}


\begin{figure}
\centering
\begin{minipage}{.5\textwidth}
  \centering

\tikzset{every picture/.style={line width=0.75pt}} 

\begin{tikzpicture}[x=0.75pt,y=0.75pt,yscale=-1,xscale=1]

\draw   (51,120) .. controls (51,115.58) and (54.58,112) .. (59,112) -- (191,112) .. controls (195.42,112) and (199,115.58) .. (199,120) -- (199,144) .. controls (199,148.42) and (195.42,152) .. (191,152) -- (59,152) .. controls (54.58,152) and (51,148.42) .. (51,144) -- cycle ;
\draw   (40,97) .. controls (40,88.72) and (46.72,82) .. (55,82) -- (194,82) .. controls (202.28,82) and (209,88.72) .. (209,97) -- (209,142) .. controls (209,150.28) and (202.28,157) .. (194,157) -- (55,157) .. controls (46.72,157) and (40,150.28) .. (40,142) -- cycle ;
\draw   (20.5,48.2) .. controls (20.5,31.52) and (34.02,18) .. (50.7,18) -- (203.8,18) .. controls (220.48,18) and (234,31.52) .. (234,48.2) -- (234,138.8) .. controls (234,155.48) and (220.48,169) .. (203.8,169) -- (50.7,169) .. controls (34.02,169) and (20.5,155.48) .. (20.5,138.8) -- cycle ;
\draw   (31,73.42) .. controls (31,61.04) and (41.04,51) .. (53.42,51) -- (197.58,51) .. controls (209.96,51) and (220,61.04) .. (220,73.42) -- (220,140.68) .. controls (220,153.06) and (209.96,163.1) .. (197.58,163.1) -- (53.42,163.1) .. controls (41.04,163.1) and (31,153.06) .. (31,140.68) -- cycle ;

\draw (60,124) node [anchor=north west][inner sep=0.75pt]   [align=left] {$\displaystyle \beta ^{k}$ not manipulable};
\draw (89,87) node [anchor=north west][inner sep=0.75pt]   [align=left] {$\displaystyle \beta ^{k}$ stable};
\draw (84,58) node [anchor=north west][inner sep=0.75pt]   [align=left] {$\displaystyle GS^{k}$ stable};
\draw (55,27) node [anchor=north west][inner sep=0.75pt]   [align=left] {$\displaystyle GS^{k}$ not manipulable};

\end{tikzpicture}

\captionof{figure}{Set inclusion of problems for $ GS^{k}$ and $ \beta^{k}$.} \label{Fig-GS-BM}
\end{minipage}%
\begin{minipage}{.5\textwidth}
  \centering

\tikzset{every picture/.style={line width=0.75pt}} 

\begin{tikzpicture}[x=0.75pt,y=0.75pt,yscale=-1,xscale=1]

\draw   (41,208) .. controls (41,203.58) and (44.58,200) .. (49,200) -- (124,200) .. controls (128.42,200) and (132,203.58) .. (132,208) -- (132,232) .. controls (132,236.42) and (128.42,240) .. (124,240) -- (49,240) .. controls (44.58,240) and (41,236.42) .. (41,232) -- cycle ;
\draw   (33.9,202) .. controls (33.9,195.37) and (39.27,190) .. (45.9,190) -- (229,190) .. controls (235.63,190) and (241,195.37) .. (241,202) -- (241,238) .. controls (241,244.63) and (235.63,250) .. (229,250) -- (45.9,250) .. controls (39.27,250) and (33.9,244.63) .. (33.9,238) -- cycle ;
\draw   (25,153.4) .. controls (25,140.48) and (35.48,130) .. (48.4,130) -- (118.6,130) .. controls (131.52,130) and (142,140.48) .. (142,153.4) -- (142,235.7) .. controls (142,248.62) and (131.52,259.1) .. (118.6,259.1) -- (48.4,259.1) .. controls (35.48,259.1) and (25,248.62) .. (25,235.7) -- cycle ;
\draw   (17,149.42) .. controls (17,133.17) and (30.17,120) .. (46.42,120) -- (222.58,120) .. controls (238.83,120) and (252,133.17) .. (252,149.42) -- (252,237.68) .. controls (252,253.93) and (238.83,267.1) .. (222.58,267.1) -- (46.42,267.1) .. controls (30.17,267.1) and (17,253.93) .. (17,237.68) -- cycle ;

\draw (50,210) node [anchor=north west][inner sep=0.75pt]   [align=left] {$\displaystyle GS^{k}$ stable};
\draw (150,199) node [anchor=north west][inner sep=0.75pt]   [align=left] {$\displaystyle GS^{k}$ not\\manipulable};
\draw (40,149) node [anchor=north west][inner sep=0.75pt]   [align=left] {$\displaystyle GS^{k+1}$ stable};
\draw (150,139) node [anchor=north west][inner sep=0.75pt]   [align=left] {$\displaystyle GS^{k+1}$ not\\manipulable};

\end{tikzpicture}

  \captionof{figure}{Set inclusion of problems for $ GS^{k}$ and $ GS^{k+1}$.}
\end{minipage}
\end{figure}




The manipulation strategy under the constrained GS is to include an unlisted acceptable school in the list. But when the constrained GS is stable, all the seats of such a school are assigned to higher-priority students, and such manipulation does not help. This implies that constrained the serial dictatorship mechanism is non-manipulable and stable for the same set of problems. 

\begin{proposition} \label{Prop-SDk}
Let $ (P,\succ,q)$ be a problem and $ k>1$. (i) The constrained serial dictatorship mechanism $ SD^{k}$ is stable if and only if it is not manipulable. (ii) If the constrained First-Preference-First mechanism $ FPF^{k}$ is stable, then the constrained Gale-Shapley mechanism $ GS^{k}$ is not manipulable.
\end{proposition}

Part (ii) of the above theorem is a surprising interplay between the two concepts for the compared mechanisms. Note that the constrained First-Preference-First mechanism and the constrained Gale-Shapley mechanism are not comparable via manipulability \citep{bonkoungou2020} and via fairness by stability (Example \ref{ex-fpf}). However, if at some profile $(P,\succ,q)$, $FPF^k$ is stable (while $GS^k$ might not), then $GS^k$ is not manipulable at $(P,\succ,q)$.

In general, the constrained Gale-Shapley mechanism may be unstable while not manipulable. We illustrate this in the following example. 

\begin{example}
Let $I=\{i_{1},\hdots,i_{4}\}$ and $ S=\{s_{1},\hdots,s_{4}\}$. Let $ (P,\succ,q)$ be a problem where each school has one seat and the remaining components are specified as follows.
\begin{center}
\begin{tabular}{c c c c|c c c c}
 $ P_{i_{1}}$ & $ P_{i_{2}}$ & $ P_{i_{3}}$ & $ P_{i_{4}}$ & $ \succ_{s_{1}}$ & $ \succ_{s_{2}}$ & $ \succ_{s_{3}}$ & $ \succ_{s_{4}}$\\\hline
$ s_{1}$ & $ s_{1}$ & $ s_{2}$ & $ s_{3}$ & $ i_{1}$ & $ i_{4}$ & $ i_{3} $ & $ \vdots$ \\
$ \vdots$ & $s_{2} $ & $ s_{3}$ & $ s_{2}$ & $ \vdots$ & $ i_{3}$ & $ i_{2}$ & \\
$ $ & $s_{3} $ & $ \vdots$ & $ \vdots$ & $ $ & $ i_{2}$ & $ i_{4}$ & \\
$  $ & $\emptyset $ & $  $ & $ $ & $  $ & $ i_{1}$ & $ i_{1}$ & \\

\end{tabular} \label{example5}
\end{center}
Let us consider the constrained Gale-Shapley mechanism $ GS^{2}$. We have
 \begin{equation*}
 GS^{2}(P,\succ,q)= \begin{pmatrix} i_{1} & i_{2} & i_{3} & i_{4} \\ s_{1} & \emptyset & s_{2} & s_{3}
 \end{pmatrix}.
 \end{equation*}
 This matching is not stable at $ (P,\succ,q)$ because student $ i_{2}$ is unmatched, finds school $ s_{3}$ acceptable while student $ i_{4}$ is matched to it and $ i_{2}\mathrel{\succ_{s_{3}}}i_{4}$. We claim that $ GS^{2}$ is not manipulable at $ (P,\succ,q)$. Only student $ i_{2}$ could benefit from misrepresenting her preferences to the mechanism $GS^{2} $ because each of the other students is matched to her most preferred school. Let $ P_{i_{2}}^{s_{3}}$ be a preference relation where student $ i_{2}$ has ranked only school $ s_{3}$ acceptable. Then,
\begin{equation*}
 GS^{2}(P_{i_{2}}^{s_{3}}, P_{-i_{2}},\succ,q)= \begin{pmatrix} i_{1} & i_{2} & i_{3} & i_{4} \\ s_{1} & \emptyset & s_{3} & s_{2}
 \end{pmatrix},
 \end{equation*} 
that is, student $ i_{2}$ remains unmatched even by ranking school $ s_{3}$ first. (It is easy to verify that any other strategy also leaves $i_2$ unmatched.) Therefore, $ GS^{2}$ is not manipulable at $ (P,\succ,q)$. The intuition is that this ranking initiates a chain of rejections which returns to this student. Student $ i_{2}$ becomes a so-called ``interrupter'' when she ranks school $ s_{3}$ first \citep{kesten2010}.
\end{example}

The contrasting result between the constrained Boston mechanism and the constrained Gale-Shapley can be traced back to the immediate versus deferred acceptance features and the constraint. Let us first state an important property of the Boston mechanism at the origin of its manipulability.

\begin{lemma}\label{lemma1} Let $k> 1$ and $(P,\succ,q)$ a problem. Student $ i$ is a manipulating student of the constrained Boston mechanism $\beta^k$ at the problem $(P,\succ,q)$ if and only if there is a school $s$ such that $s\mathrel{P_i}\beta_i^k(P,\succ,q)$, and there are less than $q_s$ number of students who ranked school $s$ first and have higher priority than student $i$ at school $s$. 
\end{lemma}

This lemma characterizes the set of manipulating students of the Boston mechanism by those who missed schools that are not ranked first by \enquote{enough} students who have higher priority than $i$ at the school in question. Clearly, if student $i$ is a blocking student of $\beta^k(P,\succ,q)$, then there is a school $s$ such that $s\mathrel{P_i}\beta_i^k(P,\succ,q)$ and a student $ j$ who is matched to school $s$ under $\beta^k(P,\succ,q)$ and has lower priority than $i$ at $s$ (or school $s$ has an unassigned seat under $\beta^k(P;\succ,q)$). Then, there are less than $q_s$ number of students who ranked school $s$ first and have higher priority than $i$ at $s$. Otherwise, student $j$ would not have been matched to school $s$ under $\beta^k(P,\succ,q)$ or no seat would have been left unassigned. This is the reason why the set of manipulating students includes blocking students.

By the deferred acceptance feature, students who are matched are neither blocking students nor manipulating students of the constrained Gale-Shapley mechanism. Blocking students and manipulating students are all students for whom the constraint is binding. That is, they are unmatched and have more acceptable schools than the ranking constraint. Because of the deferred acceptance feature, the priorities matter when a student contemplates a manipulation. No priority is violated for schools ranked within the constraint. If no student's priority is violated in any school that she has missed, then she cannot obtain a seat at any of these schools by ranking it within the constraint. To understand this conclusion, note that for any unmatched student who replaces one acceptable school with a school where her priority is not violated, the original matching remains stable under the new problem. Since she is unmatched in a stable matching, she will be unmatched in any other stable matching of the new problem \citep{roth1986}. Then, all manipulating students are also blocking students.       



 \section{Conclusions}
 
In response to various concerns, many school districts around the world have recently reformed their admissions systems. The reforms are essentially two major changes. First, they replaced the immediate acceptance procedure ( Boston mechanism) with Gale-Shapley's student-proposing deferred acceptance procedure while maintaining ranking constraints. Second, some school districts kept using the Gale-Shapley mechanism but extended the number of schools that each student is allowed to report. Anecdotal evidence points to the vulnerability to manipulation and fairness as reasons for these reforms. We showed that the immediate acceptance procedure has weakly more manipulating students than blocking students and the reverse for the Gale-Shapley's. We demonstrated that extending ranking constraints in the Gale-Shapley mechanism led to fewer blocking students. 

The fact that constrained Gale Shapley has relatively more blocking students (than it has manipulating students) suggests that in theory fairness of this mechanism might be of stronger concern than manipulability. Simultaneously, relaxing the constraint in this mechanism is guaranteed to decrease the number of blocking students. Whether this theoretical coincidence had any practical relevance for the reforms under consideration is an open question and requires further empirical research.

\bibliographystyle{apalike}
\bibliography{biblio}

\newpage

\section*{Appendix: Proofs}

To simplify the exposition we divide the appendix into three subsections. In each subsection, we order the results in logical order. All mechanisms that we consider are individually rational. We only consider blocking pairs to check for (the violation of) stability. We first present a useful lemma.

\begin{lemma}[Rural hospital theorem, \citeauthor{roth1986}, \citeyear{roth1986}] Given a problem, let $ \nu$ and $ \mu$ be two stable matchings. Then, 

(i) the same set of students are matched under $ \nu$ and $ \mu$, and 

(ii) each school is matched to the same number of students under $ \nu$ and $ \mu$, and every school which has an empty seat at one stable matching is matched to the same set of students under all stable matchings. \label{rht}

\end{lemma}

 \subsection*{Appendix A: Proof of Lemma \ref{lemma1}, and Proposition \ref{thm1}, \ref{Prop-SDk}}
 
\hspace{0.5cm}
\begin{proof}[Proof of Lemma \ref{lemma1}]
  Suppose that student $i$ is a manipulating student of $\beta^k$ at $(P,\succ,q)$. Then there is $P'_i$ such that 
  \begin{equation}\label{equation1}
      \beta_i^k(P'_i,P_{-i},\succ,q)\mathrel{P_i}\beta_i^k(P,\succ,q).
  \end{equation}
  Since $\beta^k$ is individually rational, then $\beta_i^k(P,\succ,q)\mathrel{R_i}\emptyset$ and together with Equation \ref{equation1}, we have $\beta_i^k(P'_i,P_{-i},\succ,q)=s$ for some school $s\in S$. Then student $ i$ did not rank school $s$ first under $P_i$. Suppose that there are at least $q_s$ students who ranked school $s$ first and have higher priority than student $i$ under $(P,\succ,q)$. Then all seats of school $s$ would have been allocated in the first step of the immediate acceptance algorithm to some students who ranked school $s$ first and have higher priority than student $i$. Then,  $\beta_i^k(P'_i,P_{-i},\succ,q)\neq s$, which is a contradiction. Therefore, there are less than $q_s$ number of students who ranked school $s$ first and have higher priority than student $i$ at school $s$.

  Let $(P,\succ,q)$ be a problem, $ i$ a student, and $s$ a school, and suppose that $s\mathrel{P_i}\beta^k(P,\succ,q)$ and that there are less than $q_s$ students who ranked school $s$ first and have higher priority than student $i$ at school $s$. Thus $\beta^k(P_i^s,P_{-i},\succ,q)=s$ and that student $i$ is manipulating student of $\beta^k$ at $(P,\succ,q)$. 
\end{proof}
\begin{proof}[Proof of Proposition \ref{Prop-SDk}] We call on to two claims.
\begin{claim}\label{claim1}
Suppose that student $i $ is matched to school $ s$ under $ GS^{k}(P,\succ,q)$ and let $ P_{i}^{s}$ be a preference relation where she has ranked only school $ s$ as acceptable. Then student $i $ is matched to school $s$ under $ GS^{k}(P^{s}_{i},P_{-i},\succ,q)$.
\end{claim} 

By \cite{roth1982}, $ GS_{i}(P^{k},\succ,q)=s$ implies that $GS_{i}(P^{s}_{i},P^{k}_{-i},\succ,q) =s$. We know that $(P_i^s)^k=P_i^s$. Thus, $ GS_{i}^{k}(P^{s}_{i},P_{-i},\succ,q)=s$.

\begin{claim}[\citeauthor{pathak2013}, \citeyear{pathak2013}]\label{claim2}
Suppose that student $ i$ is a manipulating student of $ GS^{k}$ at $ (P,\succ,q)$. Then she is unmatched under $ GS^{k}(P,\succ,q)$.
\end{claim} 

\textit{Part (i)}: This part is a direct corollary of part (iii) of Theorem \ref{theorem5}. 

\textit{Part (ii)}: Suppose that $ \mu=FPF^{k}(P,\succ,q)$ is stable at $ (P,\succ,q)$. By Claim \ref{claim2} every matched student under $ GS^{k}(P,\succ,q)$ is not a manipulating student of $ GS^{k}$ at $ (P,\succ,q)$. It is enough to show that no unmatched student under $ GS^{k}(P,\succ,q)$ has a profitable misrepresentation. Because $ GS^{k}$ is individually rational, by Claim \ref{claim1}, we further need to restrict ourselves to manipulation by top-ranking schools. Since $ \mu$ is stable at $ (P,\succ,q)$, we claim that it is also stable at $ (P^{k},\succ,q)$. Since $ GS^{k}$ is individually rational, we need to check that there is no blocking pair. Suppose, to the contrary, that a pair $ (i,s)$ is a blocking pair for $ \mu$ under $ (P^{k},\succ,q)$. Then, $ s\mathrel{P_{i}^{k}} \mu(i)$ and either (i) school $ s$ has an empty seat under $ \mu$ or (ii) there is a student $ j$ such that $ \mu(j)=s$ and $ i\mathrel{\succ_{s}}j$. Note that $ s\mathrel{P^{k}_{i}}\mu(i)$ implies that $ s\mathrel{P_{i}} \mu(i)$. Therefore, $ (i,s)$ is also a blocking pair for $ \mu$ under $ (P,\succ,q)$, thus contradicting our assumption that $\mu$ is stable at $(P,\succ,q)$. Therefore $ \mu$ is stable at $ (P^{k},\succ,q)$. Since $ GS(P^{k},\succ,q)$ is the student-optimal stable matching under $ (P^{k},\succ,q)$,
\begin{equation}
  \text{for each student}\:\: i,\: GS_{i}(P^{k},\succ,q) \mathrel{R_{i}^{k}}\mu(i). \label{eq2}
\end{equation}
By Lemma \ref{rht} the same set of students are matched under $ \mu$ and $ GS(P^{k},\succ,q)$. Let $ i$ be a student and $ s$ a school and suppose that $ i$ is unmatched under $GS(P^{k},\succ,q) $ and that $ s \mathrel{P_{i}} GS_{i}(P^{k},\succ,q)$. Then, student $ i$ is also unmatched under $ \mu$. Thus, $ s\mathrel{P_{i}}\mu(i)=\emptyset$. Because $ \mu$ is stable at $ (P,\succ,q)$ every student in $\mu^{-1}(s) $ has higher priority than $ i$ under $ \succ_{s}$. Let $ P^{s}_{i}$ denote a preference relation where $ i$ has ranked only school $ s$ acceptable. Since $ \mu$ is stable at $ (P^{k},\succ,q)$ it is also stable at $ (P^{s}_{i},P^{k}_{-i},\succ,q)$. By Lemma \ref{rht}, the set of matched students is the same at all stable matchings. Thus, student $ i$ is also unmatched under $ GS(P^{s}_{i},P^{k}_{-i},\succ,q)$. By Claim \ref{claim1}, there is no preference relation $ P'_{i}$ such that $ GS^{k}_{i}(P'_{i},P_{-i})=s$. Thus, $ GS^{k}$ is not manipulable at $ (P,\succ,q)$.

\end{proof}

\begin{proof}[Proof of Proposition \ref{thm1}] \textit{Part (i)}: The Boston mechanism is a special case of the First-Preference-First mechanism when every school is a first-preference-first school. Suppose that $ \beta^{k}(P,\succ,q)$ is stable at $ (P,\succ,q)$. By equation \ref{eq2}, each student finds the outcome $ GS^{k}(P,\succ,q)$ at least as good as $ \beta^{k}(P,\succ,q)$ under $ P^{k}$. We also know that the Boston mechanism is Pareto efficient, that is, for each problem, there is no other matching that each student finds at least as good as its outcome \citep{abdul2003}. Therefore, the matching $ \beta^{k}(P,\succ,q)=\beta(P^{k},\succ,q)$ is Pareto efficient under $ P^{k}$. Thus, 
$ GS^{k}(P,\succ,q)=\beta^{k}(P,\succ,q)$ is stable at $ (P,\succ,q)$.

We construct a problem where $ GS^{k}$ is stable but not $ \beta^{k}$. Since there are at least two schools and more students than schools, let $ s_{1},s_{2}$ be two distinct schools and $ i_{1}, i_{2}$ and $i_{3}$ three students. Let $ (P,\succ,q)$ be a problem where each school has one seat and the remaining components are specified as follows.

\begin{center}
\begin{tabular}{c c|c}
$ P_{i\neq 3}$ & $ P_{3}$ & $ \succ_{s\in S}$ \\\hline
 $s_{1}$ & $s_{2}$ & $i_{1}$\\ $s_{2}$ & $ s_{1}$ & $ i_{2}$ \\ $\emptyset$ & $\emptyset$& $i_{3}$
 \\ $ $ & $ $ & $\vdots$
\end{tabular}
\end{center}
Since $ k\geq 2$, $ GS^{k}(P,\succ,q)=GS(P,\succ,q)$ is stable at $ (P,\succ,q)$. However, the matching
\begin{equation*}
\beta^{k}(P,\succ,q)=\begin{pmatrix} i_{1} & i_{3} & i\neq 1,3\\ s_{1} & s_{2} & \emptyset 
\end{pmatrix}
\end{equation*}
is not stable because the pair $ (i_{2},s_{2})$ blocks it under $ (P,\succ,q)$.

\textit{Part (ii)}: This part is a direct corollary of Theorem \ref{thm1}.
\end{proof}

$ $

\subsection*{Appendix B: Proof of Theorems \ref{Th-GS-counting}}

\hspace{0.9cm}

\begin{lemma} \label{lemme-base} Let $ N$ be a subset of students and $\mu=GS(P^{\ell}_{N},P^{k}_{-N},\succ,q)$. Any blocking student for $ \mu$ under $ (P,\succ,q)$ is unmatched. 
\end{lemma}

\begin{proof}
We prove it by the contradiction. Suppose, to the contrary, that student $ i$ is a blocking student for $\mu$ under $ (P,\succ,q)$ such that $ \mu(i)=s$ for some school $ s$. Then, there is a school $ s'$ such that $ s' \mathrel{P_{i}}\mu(i)$ and either (i) $ \lvert\mu^{-1}(s')\lvert < q_{s'}$ or (ii) there is a student $ j$ such that $ \mu(j)=s'$ and $ i\mathrel{\succ_{s'}}j$. Since $ \mu(i)=s$, school $ s$ is one of the top $ x$ acceptable schools under $ P_{i} $ where $x=\ell$ if $i\in N$ and $x=k$ if $x\notin N$. Thus $ s' \mathrel{P^{x}_{i}} \mu(i)=s$ and $ (i,s')$ is a blocking pair of $ \mu$ under $ (P^{\ell}_{N},P^{k}_{-N},\succ,q)$, contradicting the stability of $\mu$ under $ (P^{\ell}_{N},P^{k}_{-N},\succ,q)$. 

\end{proof}

\begin{proof}[Proof of Theorem \ref{Th-GS-counting}] We call on to the sequential version of \cite{mcvitie1970} of the deferred acceptance algorithm. This is a version where students apply one at a time according to a predetermined order such that in each step the highest-ordered student among the ones whose applications have not yet been tentatively accepted applies. 

The idea of the proof is to consider close ranking constraints $k-1$ and $k$, where $k>1$, and replace students' preference relations in $P^{k-1}$ with the ones in $P^k$. Let $N\subsetneq I$ be a proper subset of $I$ and $i\notin N$. Suppose that starting from $P^{k-1}$ we have replaced all the preferences of students in $N$ by their preferences in $P^{k}$ and define $ \widehat{P}=(P_N^k,P^{k-1}_{-N})$. Note that $N$ may be empty. Next we consider student $i$: from $\widehat{P}$ we replace her preference relation $P^{k-1}_i$ by $P^k_i$. Let $ \nu=GS(P_i^{k-1},\widehat{P}_{-i},\succ,q)$ and $\mu=GS(P^k_i,\widehat{P}_{-i},\succ,q)$. 

If student $i$ has less than $k$ schools acceptable under $P_i$ or is matched under $\nu$, of course to one of her top $k-1$ most preferred schools, then $\mu=\nu$ and both matchings have the same number of blocking students. Without loss of generality suppose that student $i$ has at least $k$ schools acceptable under $P_i$ and $\nu(i)=\emptyset$. Note that student $i$ has applied to all her first $ k-1$ acceptable schools in the algorithm for $\nu$. Let $ s$ denote her $ k$'th acceptable school under $ P^k_{i}$. The algorithm for $ \mu$ is a continuation of the one for $ \nu$ by letting $ i$ apply to school $ s$ and completing the subsequent sequences of applications and rejections. Consider the following possible pointing sequences starting from the step at which student $ i$ applied to school $ s$: 

\begin{equation}\label{seq1}
    i \rightarrow s \rightarrow i_1 \rightarrow s_1 \rightarrow i_2 \rightarrow s_2 \rightarrow ... i_n \rightarrow s_n \rightarrow i_{n+1} \rightarrow \emptyset, 
\end{equation}

\begin{equation}\label{seq2}
    i \rightarrow s \rightarrow i_1 \rightarrow s_1 \rightarrow i_2 \rightarrow s_2 \rightarrow ... i_n \rightarrow s_n \rightarrow i_{n+1} \rightarrow s_{n+1} 
\end{equation} where the pointing $ i'\rightarrow s'$ means that student $ i'$ applies to school $s'$ and $ s'\rightarrow i'$ means that school $s'$ rejects the application of student $i'$. In the sequence in equation \ref{seq1}, the pointing $ i_{n+1}\rightarrow \emptyset$ means that student $i_{n+1}$ applied to all her acceptable schools and thus remained unmatched. In the sequence in equation \ref{seq2}, the pointing $ i_{n+1}\rightarrow s_{n+1}$ for which school $ s_{n+1}$ does not point to any student means that school $ s_{n+1}$ did not reject any student after $ i_{n+1}$'s application. Note that there might be cycles and some students could appear several times in each sequence. 

Let $ \alpha$ and $ \beta$ denote the number of blocking students of $ \mu$ and $ \nu$ respectively, among the students in $ I\setminus \{i,i_{1},\hdots,i_{n+1}\}$ and $x$ and $ y$ the number of blocking students of $\mu$ and $\nu$ respectively, among the students in $ \{i,i_{1},\hdots,i_{n+1}\}$.

\begin{claim*}
$ \beta \geq \alpha$, and [$ x=y $ or $ x=y-1$].
\end{claim*}
We first prove that $ x=y $ or $ x=y-1$. Note first that all students in $\{i,i_{1},\hdots,i_{n+1}\}\setminus \{i,i_{n+1}\}$ are matched under $\mu$ and $\nu$. By Lemma \ref{lemme-base} they are not blocking students of $\mu$ and $\nu$. Hence, the comparison of the number of blocking students of $\mu$ and $\nu$ among the students in $ \{i,i_1,\hdots,i_{n+1}\}$ concerns students $i$ and $i_{n+1}$. We consider two cases:

Case 1: $i=i_{n+1}$. 
If student $i$ is not a blocking student of $\nu$, then school $s$ does not have an empty seat under $\nu$ and for each student $j\in\nu^{-1}(s)$, $j\mathrel{\succ_s} i$. The sequence is the one in equation \ref{seq1} but in simple version $i\rightarrow s \rightarrow i \rightarrow \emptyset$. Thus $\mu=\nu$ and student $i$ is not a blocking student of $\mu$. Thus $x=y$.

Case 2: $i\neq i_{n+1}$. In this case $i_1 \neq i$. We claim that student $i$ is a blocking student of $\nu$ but not a blocking student of $\mu$. The reason is that school $s$ has accepted $i$'s application and rejected student $i_1$. Thus $i\mathrel{\succ_s}i_1$ and $\mu(i)=s$. Since $\nu(i)=\emptyset$ and $\nu(i_1)=s$, student $i$ is a blocking student of $\nu$. Since $\mu(i)=s$, by Lemma \ref{lemme-base}, she is not a blocking student of $\mu$. Note now that student $i_{n+1}$ is matched under $\nu$. Thus by Lemma \ref{lemme-base} she is not a blocking student of $\nu$. Therefore, $0=x=y-1$ if $i_{n+1}$ is not a blocking student of $\mu$ and $1=x=y$ if $i_{n+1}$ is a blocking student of $\mu$. Overall, $ x=y $ or $ x=y-1$.

We now prove that $ \beta \geq \alpha$. We prove the claim that there is no student in $I\setminus \{i,i_1,\hdots,i_{n+1}\}$ who is not a blocking student of $\nu$ but a blocking student of $\mu$. Let $j\notin I\setminus \{i,i_1,\hdots,i_{n+1}\}$ and note that $ \nu(j)=\mu(j)$. If student $j$ is matched under $\nu$, then she is not a blocking student of $\nu$ and $ \mu$. Suppose that $ \nu(j)=\emptyset$, and for a contradiction, that she is not a blocking student of $\nu$ but a blocking student of $\mu$ under $(P,\succ,q)$. There is a school $s' $ such that $s' \mathrel{P_j} \mu(j)$ and either (i) school $s'$ has an empty seat under $ \mu$ or (ii) there is a student $ j'$ such that $ \mu(j')=s'$ and $ j \mathrel{\succ_{s'}j'}$. Consider (i). As above, school $s'$ has an empty seat under $\nu$. Therefore, student $j$ is also a blocking student of $\nu$, contradicting our assumption. Consider (ii). Suppose that school $ s'$ did not reject any student in any step from the step at which student $ i$ applies to school $ s$ to the end. Then $ \mu(j')=\nu(j')=s'$ and student $j$ is also a blocking student of $\nu$. Suppose that school $ s'$ has rejected a student $ j''$ such that $\nu(j'')=s'$. Since $\mu(j')=s' $, then $ j'\mathrel{\succ_{s'}} j''$. Thus $ j\mathrel{\succ_{s'}} j''$ and since $ \nu(j'')=s'$, student $j$ is a blocking student of $ \nu$, contradicting again our assumption. Therefore there is no student who is simultaneously not a blocking student of $\nu$ but a blocking student of $\mu$ under $(P,\succ,q)$. Thus $ \alpha \geq \beta$. 

By this claim, $ \alpha+ y \geq \beta + x$. Therefore, $ \nu$ has a weakly larger number of blocking students than $ \mu$. Starting from $N=\emptyset$ and successively replacing students' preferences in any order we conclude that $GS(P^{k-1},\succ,q)$ has weakly more number of blocking students than $GS(P^{k},\succ,q)$.

Finally, we describe a problem where the outcome of $ GS^{\ell}$ has more blocking students than the outcome of $ GS^{k}$. Let $ (P,\succ,q)$ be a problem where each school has one seat, each student has $ k$ acceptable schools, and such that students have a common ranking of schools. Then, $ GS^{k}(P,\succ,q)=GS(P,\succ,q)$. Thus $ GS^{k}(P,\succ,q)$ is stable at $ (P,\succ,q)$. Let $ s$ be the school that students have ranked at the $ k$'th position starting from the top. Since there are more students than schools and $ k>\ell$, at least one student is not matched under $ GS^{\ell}(P,\succ,q)$ and no student is matched to school $ s$ even though every student prefers it to be unmatched. Therefore, there are more blocking students for $ GS^{\ell}(P,\succ,q)$ than $ GS^{k}(P,\succ,q)$ under $ (P,\succ,q)$.
\end{proof}


$ $

\subsection*{Appendix C: Proof of Theorem \ref{theorem5}}

\hspace{0.9cm}

\begin{proof}[Proof of Theorem \ref{theorem5}]
\textit{Part (i)}: Let $ i$ be a blocking student of $ \mu=\beta(P^{k},\succ,q)$. There is a school $ s$ such that the pair $ (i,s)$ blocks $ \mu$ under $ (P,\succ,q)$. Then, $ s\mathrel{P_{i}}\mu(i)$ and either (a) school $ s$ has an empty seat under $ \mu$ or (b) there is a student $ j$ such that $ \mu(j)=s$ and $ i\mathrel{\succ_{s}}j$. We claim that student $ i$ did not rank school $ s$ first under $ P_{i}$. Otherwise, school $s$ has rejected student $ i$ at the first step of the Boston algorithm under $ (P^{k},\succ,q)$. This is because $ k>1$ and the top-ranked schools are considered under $ \beta^{k}$. This contradicts the assumption that school $ s$ has an empty seat or has accepted student $ j$ with $ i\mathrel{\succ_{s}}j$. Let $ P^{s}_{i}$ be a preference relation where $ i$ has ranked school $ s$ first. Since $ s$ has an empty seat under $ \beta^{k}(P,\succ,q)$ or has accepted student $ j$ with $ i\mathrel{\succ_{s}}j$, there are less than $ q_{s}$ students who have ranked school $ s$ first under $ P^{k}$ and have a higher priority than $ i$ under $ \succ_{s}$. Therefore, $ \beta_{i}^k(P^{s}_{i},P_{-i},\succ,q)=s$. Since $ s\mathrel{P_{i}}\mu(i)$, $ i$ is a manipulating student of $ \beta^{k}$ at $ (P,\succ,q)$.

\textit{Part (ii)}: We prove this part by contradiction. Suppose that student $ i$ is a manipulating student of $ GS^{k}$ at $ (P,\succ,q)$ but is not a blocking student of $\mu=GS^{k}(P,\succ,q)$ under $ (P,\succ,q)$. By Claim \ref{claim2}, $ i$ is unmatched under $ GS^{k}(P,\succ,q)$. Let $ s$ be a school such that $ s\mathrel{P_{i}} \mu(i)$. Then, $ \lvert\mu^{-1}(s)\lvert=q_{s}$ and every student in $\mu^{-1}(s) $ has higher priority than $ i$ under $ \succ_{s}$. Let $ P^{s}_{i}$ be a preference relation where $ i$ has ranked only school $ s$ as an acceptable school. Since $ \mu$ is stable at $ (P^{k},\succ,q)$, it is also stable at $ (P^{s}_{i},P^{k}_{-i},\succ,q)$. This follows from the fact that $ \mu(i)=\emptyset$ and that every student in $ \mu^{-1}(s)$ has higher priority than $ i$ under $ \succ_{s}$. By Lemma \ref{rht}, the set of unmatched students is the same under $ \mu$ and $ GS(P^{s}_{i},P^{k}_{-i},\succ,q) $. Thus, $ i$ is also unmatched under $ GS^{k}(P^{s}_{i},P_{-i},\succ,q)$. By Claim \ref{claim1}, there is no misreport by which student $ i$ is matched to $ s$. Since $ s$ has been chosen arbitrarily, $ i$ is not a manipulating student of $ GS^{k}$ at $ (P,\succ,q)$, a conclusion that contradicts our assumption. Therefore student $i$ is a blocking student of $GS^{k}(P,\succ,q)$ under $(P,\succ,q)$.

\textit{Part (iii)}. By part (ii) every manipulating student of $SD^k$ is a blocking student of $SD^k(P,\succ,q)$. Let $i$ be a blocking student of $\mu=SD^k(P,\succ,q)$. Then at $i$'s turn, there is no seat left among her top $k$ acceptable schools. She is then left unmatched while she has ranked a school $s$ as acceptable and below the position $k$ which still has a seat available. Let $P^s_{i}$ be a preference relation where $i$ ranks $s$ first. Then $SD^k_i(P^s_i,P_{-i},\succ,q)=s$. Therefore, student $i$ is a manipulating student of $SD^k$ at $(P,\succ,q)$.
\end{proof}

\end{document}